\documentclass[11pt]{article}
\usepackage{amsmath, amssymb, amsfonts, amsthm}

\newtheorem{theorem}{Theorem}[section]
\newtheorem{proposition}[theorem]{Proposition}

\newcommand{\rd}{{\rm d}}
\newcommand{\be}{\begin{equation}}
\newcommand{\ee}{\end{equation}}
\newcommand{\bey}{\begin{eqnarray}}
\newcommand{\eey}{\end{eqnarray}}

\newcommand{\bx}{{\bf x}}

\newcommand{\ph}{\varphi}

\renewcommand{\a}{\alpha}

\newcommand{\cU}{{\cal U}}

\newcommand{\bR}{{\mathbb R}}

\newcommand{\bN}{{\mathbb N}}

\newcommand{\tr}{\mbox{Tr}}

\newcommand{\const}{\mathrm{const}}

\newcommand{\cB}{{\cal B}}

\newcommand{\cK}{{\cal K}}
\newcommand{\cH}{{\cal H}}

\input epsf

\newcommand{\donothing}[1]{}

\begin{document}

\title{Quantum Dynamics with Mean Field Interactions: \\a New Approach}
\author{L\'aszl\'o Erd\H os\thanks{Partially supported by SFB/TR12 Project from DFG} \; , Benjamin Schlein\thanks{Supported by a Kovalevskaja Award from the Humboldt Foundation.
On leave from Cambridge University}
\\
\\
Institute of Mathematics, University of Munich, \\
Theresienstr. 39, D-80333 Munich, Germany}

\maketitle

\begin{center}
{\it {\large Dedicated to J\"urg Fr\"ohlich on the occasion of his 60th birthday, \\ with admiration and gratitude}}
\end{center}

\begin{abstract}
We propose a new approach for the study of the time evolution of a factorized $N$-particle bosonic wave function with respect to a mean-field dynamics with a bounded interaction potential. The new technique, which is based on the control of the growth of the correlations among the particles, leads to quantitative bounds on the difference between the many-particle Schr\"odinger dynamics and  the one-particle nonlinear Hartree dynamics. In particular the one-particle density matrix associated with the solution to the $N$-particle Schr\"odinger equation is shown to converge to the projection onto the one-dimensional subspace spanned by the solution to the Hartree equation with a speed of convergence of order $1/N$ for all fixed times.
\end{abstract}

\section{Introduction}

We consider a system of $N$ interacting bosons in $\nu$ dimensions described on the Hilbert space $\cH_N = L^2_s (\bR^{N\nu}, \rd x_1 ,\dots \rd x_N)$, the subspace of $L^2 (\bR^{N\nu}, \rd x_1\dots \rd x_N)$ consisting of permutation symmetric wave functions. Here the variables $x_1, \dots, x_N \in \bR^{\nu}$ refer to the positions of the $N$ particles. The Hamiltonian is given by
\begin{equation}
H_N = \sum_{j=1}^N \left( -\Delta_{x_j} + U (x_j) \right) + \frac{1}{N} \sum_{i<j}^N  V (x_i - x_j)
\end{equation}
and, under suitable conditions on the potentials $U$ and $V$, acts as a self-adjoint operator on $\cH_N$. The coupling constant $1/N$ in front of the interaction characterizes mean field models; it guarantees that the kinetic and the potential part of the Hamiltonian are typically of the same order.
\medskip

The dynamics of the system is governed by the $N$-particle Schr\"odinger equation \begin{equation}\label{eq:schr} i\partial_t \psi_{N,t} = H_N \psi_{N,t} \qquad \Rightarrow \quad \psi_{N,t} = e^{-iH_N t} \psi_N \, . \end{equation}
In particular we will be interested in the time evolution of factorized initial data $\psi_N (x_1, \dots ,x_N) = \prod_{j=1}^N \ph (x_j)$. It turns out that, because of the mean field nature of the interaction, for large $N$, factorization is approximatively preserved by the evolution (\ref{eq:schr}). In a sense to be made precise, we have $\psi_{N,t} \simeq \ph_t^{\otimes N}$, for a suitable $\ph_t \in L^2 (\bR^{\nu})$. If factorization were indeed preserved, it is simple to check that $\ph_t$ must be determined by the solution to the self-consistent nonlinear Hartree equation \begin{equation}\label{eq:hartree} i\partial_t \ph_t = -\Delta \ph_t + ( V * |\ph_t|^2) \ph_t , \end{equation} with initial data $\ph_{t=0} = \ph$.

\medskip

A more precise analysis shows that the solution $\psi_{N,t}$ can be approximated, in the limit $N \to \infty$, by products of solutions to the nonlinear Hartree equation (\ref{eq:hartree}) only if we are interested in quantities depending non-trivially only on a finite number of particles. In other words, we cannot expect the $L^2$-difference between $\psi_{N,t}$ and $\ph_t^{\otimes N}$ to converge to zero. What can be expected and, for a large class of potentials, has also been rigorously proven, is that for an arbitrary fixed $k$-particle observable $J^{(k)}$ (for example, a compact operator on $L^2 (\bR^{k \nu})$), and for every fixed time $t \in \bR$,
\begin{equation}\label{eq:conv1} \langle \psi_{N,t} , (J^{(k)} \otimes 1^{(N-k)}) \psi_{N,t} \rangle \to \langle \ph_t^{\otimes N} , (J^{(k)} \otimes 1^{(N-k)}) \ph_t^{\otimes N} \rangle = \langle \ph_t^{\otimes k}, J^{(k)} \ph_t^{\otimes k} \rangle \end{equation}
as $N \to \infty$, where $\ph_t$ is the solution to (\ref{eq:hartree}) with initial data $\ph_{t=0} = \ph$. The convergence (\ref{eq:conv1}) can be interpreted as the convergence of the marginal densities. Recall that, for a given $N$-particle wave-function $\psi_N$, we define the density matrix $\gamma_N = |\psi_N \rangle \langle \psi_N|$ as the orthogonal projection onto $\psi_N$, and, for $k=1, \dots ,N$, we define the $k$-particle marginal density $\gamma^{(k)}_N$ associated with $\psi_N$ by taking the partial trace of $\gamma_N$ over the last $N-k$ particles. In other words, $\gamma_{N}^{(k)}$ is defined for $k=1, \dots ,N$, as a non-negative trace class operator on $L^2 (\bR^{\nu k})$ with kernel given by
\begin{equation}
\begin{split}
\gamma^{(k)}_{N} (\bx_k ; \bx'_k) &= \int \rd \bx_{N-k} \, \gamma_N (\bx_k, \bx_{N-k} ; \bx'_k , \bx_{N-k}) = \int \rd \bx_{N-k} \, \psi_N (\bx_k, \bx_{N-k}) \overline{\psi}_N (\bx'_k , \bx_{N-k}).
\end{split}
\end{equation}
Here and henceforth we use the notation $\bx_k = (x_1, \dots ,x_k)$, $\bx'_k = (x'_1, \dots x'_k)$, $\bx_{N-k} = (x_{k+1} , \dots ,x_N)$. We will also use the shorthand notation $\bx \equiv \bx_N = (x_1, \dots ,x_N)$. Note that we choose the convention $\tr \; \gamma^{(k)}_N = 1$ for all $N \geq 1$, and $1 \leq k \leq N$.

\medskip

For $k =1, \dots ,N$ denote by $\gamma_{N,t}^{(k)}$ the $k$-particle marginal density associated with the solution to the $N$-particle Schr\"odinger equation (\ref{eq:schr}) with factorized initial data $\psi_N = \ph^{\otimes N}$. In terms of marginal densities, (\ref{eq:conv1}) is equivalent to the convergence
\begin{equation}\label{eq:conv2}
\gamma_{N,t}^{(k)} \to |\ph_t \rangle \langle \ph_t|^{\otimes k} \qquad \text{as } N \to \infty
\end{equation}
for all fixed $t \in \bR$ and $k \geq 1$, with respect to the weak* topology defined on the space of trace class operators. It is simple to check that the weak* convergence in this case is equivalent to the convergence in the trace norm topology, because the limit is a rank one projection.

\medskip

The first rigorous results of the form (\ref{eq:conv1}) or (\ref{eq:conv2}) were obtained by Hepp in \cite{Hepp}, and then extended by Ginibre and Velo in \cite{GV}. These works are based on the analysis of the time evolution of coherent states (in a second quantized representation of the system). The use of coherent states (characterized by a non-fixed number of particles) makes it possible to isolate the main part of the evolution (described by the Hartree equation) and to study the fluctuations around it. With this method, it was recently proved in \cite{RS} that (\ref{eq:conv2}) can be improved (for example for the case $k=1$) to the quantitative estimate
\begin{equation}\label{eq:RS} \tr \Big| \gamma^{(1)}_{N,t} - |\ph_t \rangle \langle \ph_t| \Big| \leq \const \cdot \frac{e^{Ct}}{N^{1/2}} \end{equation} for factorized initial data (while for coherent initial data, the results of \cite{RS} establish a better bound, proportional to $1/N$) and for interaction potentials $V$ such that $V^2 (x) \leq C (1-\Delta_x)$.

\medskip

A different approach was proposed by Spohn in \cite{Spohn}. The main idea of the technique introduced in \cite{Spohn} consists in characterizing the limit of the marginal densities $\gamma^{(k)}_{N,t}$ through the evolution equations that they satisfy. For finite $N$, the dynamics of the marginal densities
$\{ \gamma^{(k)}_{N,t} \}_{k=1}^N$ is governed by the so-called BBGKY hierarchy
\begin{equation}\label{eq:BBGKY}
\begin{split}
i\partial_t \gamma^{(k)}_{N,t} =\; & \sum_{j=1}^k \left[ -\Delta_{x_j} , \gamma^{(k)}_{N,t} \right] + \frac{1}{N} \sum_{i<j} \left[ V (x_i -x_j) , \gamma^{(k)}_{N,t} \right] \\ &+
\frac{N-k}{N} \sum_{j=1}^k \tr_{k+1} \, \left[ V(x_j -x_{k+1}), \gamma^{(k+1)}_{N,t} \right] \,,
\end{split}
\end{equation}
where $\tr_{k+1}$ denotes the partial trace over the $(k+1)$-th particle. Using the a-priori bound $\tr \; \gamma^{(k)}_{N,t}= 1$, an abstract argument shows
that the sequence $\{ \gamma^{(k)}_{N,t} \}_{k=1}^N$ is compact with respect to a suitable weak topology. The first main step towards a proof of (\ref{eq:conv2}) in \cite{Spohn} is then to show that an arbitrary limit point $\{\gamma_{\infty,t}^{(k)} \}_{k\geq 1}$ of the sequence of marginals $\{ \gamma^{(k)}_{N,t}\}_{k=1}^N$ is a solution to the infinite hierarchy
\begin{equation}\label{eq:infhier}
\begin{split}
i\partial_t \gamma^{(k)}_{\infty,t} =\; & \sum_{j=1}^k \left[ -\Delta_{x_j} , \gamma^{(k)}_{\infty,t} \right] + \sum_{j=1}^k \tr_{k+1} \, \left[ V(x_j -x_{k+1}), \gamma^{(k+1)}_{\infty,t} \right]\,
\end{split}
\end{equation}
which is obtained from (\ref{eq:BBGKY}) by formally letting $N \to \infty$. The second main step to show (\ref{eq:conv2}) consists in proving the uniqueness of the solution to (\ref{eq:infhier}). Since it is simple to verify that the ansazt $\gamma^{(k)}_t = |\ph_t \rangle \langle \ph_t|^{\otimes k}$ is indeed a solution to (\ref{eq:infhier}) if $\ph_t$ satisfies the Hartree equation (\ref{eq:hartree}), the compactness of the sequence $\{ \gamma^{(k)}_{N,t} \}_{k=1}^N$ leads to (\ref{eq:conv2}).
This method, which was first used in \cite{Spohn} to prove (\ref{eq:conv2}) for bounded interaction, was recently extended to singular potential (see \cite{EY,BGM,ES,ABGT,AGT,ESY0,EESY,ESY1,ESY2,ESY3}). Note that this argument does not provide an effective estimate for the speed of convergence in (\ref{eq:conv2}).

\medskip

Recently, a new proof of (\ref{eq:conv2}) has been presented by Fr\"ohlich et al. in \cite{FGS}. This method provides convergence estimates uniformly in Planck's constant up to an exponentially small remainder, and therefore it can be used to combine the mean field- with the semiclassical limit (note that, in the semiclassical limit the Hartree equation converges to a Vlasov equation). In \cite{FKP} it has been shown that the limit (\ref{eq:conv2}) can be interpreted as a Egorov-type theorem, in the sense that the time evolution of observables commutes with their Wick quantization, up to corrections which vanish in the mean-field limit.

\medskip

In the present work we propose yet a different approach. With respect to the techniques mentioned above, our approach has the advantage of being very simple and of giving a quantitative
bounds of the order $1/N$ for the speed of convergence in (\ref{eq:conv2}). The estimates we obtain here are therefore better, in their $N$-dependence, than the bounds (\ref{eq:RS}) derived, for factorized initial data, in \cite{RS}. However, our analysis is restricted to a smaller class of interaction potential $V$; in fact, we need to assume that the operator norm $\| V \|$ is finite (the operator norm of the multiplication operator $V$ equals the $L^{\infty}$-norm of the function $V(x)$), and that it has a Fourier transform $\widehat{V} \in L^1 (\bR^{\nu})$. The main result of this paper is the following theorem.

\begin{theorem}\label{thm:main}
Fix an arbitrary dimension $\nu \geq 1$. Suppose that $V \in L^{\infty} (\bR^{\nu})$ is such that $\| \widehat{V} \|_1 < \infty$. Denote by $\gamma_{N,t}^{(k)}$ the $k$-particle density associated with the solution $\psi_{N,t} = e^{-iH_N t} \psi_N$ of the Schr\"odinger equation with initial data $\psi_N = \ph^{\otimes N}$, for some $\ph \in H^1 (\bR^3)$, with $\| \ph \|_2 = 1$. Then we have
\[ \tr \; \Big| \gamma_{N,t}^{(k)} - |\ph_t \rangle \langle \ph_t|^{\otimes k} \Big| \leq \frac{k^2 \lambda_V}{N} \; e^{2k \| V \| } \, \left( e^{8 \| V\| t}-1 \right) \]
for every $t \in \bR$, $N \geq 1$ and $1 \leq k \leq N$. Here we defined $\lambda_V = 1 + (\| \widehat{V} \|_1  / \| V \|)$ and $\ph_t$ is the solution to the nonlinear Hartree equation
\[ i\partial_t \ph_t = -\Delta \ph_t + (V * |\ph_t|^2) \ph_t \, , \] with $\ph_{t=0} = \ph$.
\end{theorem}

In order to prove Theorem \ref{thm:main} we control of the growth of the correlations among the particles (at time $t=0$, the wave function is factorized, and therefore there are no correlations).
Our proof has been inspired by the techniques used in \cite{NRSS} to prove Lieb-Robinson inequalities for anharmonic lattice systems (see also the proof of Lieb-Robinson bounds for quantum spin systems given in \cite{NS,NOS,LR}). To prove Lieb-Robinson bounds, one has to control the growth of the correlations between spatially separated observables. To prove Theorem \ref{thm:main}, on the other hand, we have to control the growth of correlations between observables acting on different particles.

\medskip

In Section \ref{sec:growth}, we prove that if $A$ and $B$ are two operators acting on different particles, then the commutator $[A, e^{iH_Nt} B e^{-iH_Nt} ]$ remains of order $1/N$ for all fixed times (the commutator is zero at time $t=0$); see Theorem \ref{thm:comm}.
We apply this result to prove that, if $\gamma^{(k)}_{N,t}$ denote the marginal density associated with the evolution $\psi_{N,t} = e^{-iH_N t} \psi_N$, for a factorized initial data $\psi_N = \ph^{\otimes N}$, then the difference $\gamma^{(m+n)}_{N,t} - \gamma^{(m)}_{N,t} \otimes \gamma^{(n)}_{N,t}$ remains of order $1/N$ for all fixed times (it is zero at time $t=0$),
when tested against product observables; see Proposition \ref{prop2}. Note that both Theorem \ref{thm:comm} and Proposition \ref{prop2} only require bounded potential to hold (the condition $\| \widehat{V} \|_1 < \infty$ is not required here). Finally, in Section \ref{sec:proof}, we apply the result of Proposition \ref{prop2} to show Theorem \ref{thm:main}.

\medskip

{\it Notation.} For $k \geq 1$, denote by $\cB_k$ the algebra of bounded operators over $L^2 (\bR^{k\nu})$. For an operator $A \in \cB_k$, and for integers $1 \leq i_1 < i_2 < \dots < i_k \leq N$, we denote by $A^{[i_1 ,\dots ,i_k]}$ the operator on $L^2 (\bR^{N \nu})$ acting as $A$ over the particles $i_1, \dots ,i_k$ and as the identity over all other particles. We will use the notation $\| A \|$ to indicate the operator norm of $A$ (where $A$ is meant as an operator from an $L^2$ space to itself). In particular $\| V \|$ denotes the operator norm of the potential (interpreted as a multiplication operator) which of course equals the $L^{\infty}$-norm of the function $V(x)$.

\section{Growth of Correlations}\label{sec:growth}

\begin{theorem}\label{thm:comm}
Suppose that $V \in L^{\infty} (\bR^{\nu})$. Then, for every bounded $m$-particle observable $A \in \cB_m$, and every bounded $n$-particle observable $B \in \cB_n$, we have
\begin{equation}\label{eq:claim1}
\left\| \left[ A^{[i_1, \dots , i_m]} , e^{iH_N t} B^{[j_1 , \dots, j_{n}]} e^{-i H_N t} \right] \right\| \leq \frac{m n \| A \| \, \| B \|}{N} \, \left( e^{4 \| V \| |t|} - 1\right)
\end{equation}
for all integers $1 \leq i_1 < \dots < i_m \leq N$, $1 \leq j_1 < \dots <j_n \leq N$, with
$\{ i_1, \dots ,i_m \} \cap \{ j_1, \dots ,j_n \} = \emptyset$.
\end{theorem}

\begin{proof}
Without loss of generality, we assume that $t \geq 0$. For $n,m \in \bN$, we define the quantity
\begin{equation}\label{eq:f1} f_{m,n} (t) = \sup_{A \in \, \cB_m , B \in \, \cB_n} \frac{\left\| \left[ A^{[i_1, \dots, i_m]} , e^{iH_N t} B^{[j_1, \dots , j_n]} e^{-iH_N t} \right] \right\|}{\| A \| \| B \|} \, . \end{equation}
for arbitrary integers $1 \leq i_1 < \dots < i_m \leq N$, and $1 \leq j_1 < \dots < j_n \leq N$ such that $\{ i_1, \dots ,i_m \} \cap \{ j_1, \dots , j_n \} = \emptyset$. Note that, because of the permutation symmetry of $H_N$, the quantity $f(t)$ is independent of the choice of the indices $i_{\ell}, j_{\ell}$ and
\[ f_{m,n} (t) = \sup_{A \in \, \cB_m , B \in \, \cB_n} \frac{\left\| \left[ A^{[n+1, \dots, n+m]} , e^{iH_N t} B^{[1, \dots , n]} e^{-iH_N t} \right] \right\|}{\| A \| \| B \|} \, . \]
We define next the modified Hamiltonian \[ H_N^{(n)} = H_N - \frac{1}{N} \sum_{\ell= 1}^{n} \sum_{j>n} V (x_{\ell} - x_j) \, . \] With respect to the dynamics generated by $H_N^{(n)}$, particles $1,\dots,n$ are decoupled from the rest of the system. In particular this implies that $e^{-i H_N^{(n)} t} B^{[1,\dots ,n]} e^{i H_N^{(n)} t}$ is still an operator acting only on the degrees of freedom of particles $1, \dots ,n$, with norm equal to the norm of $B$. Therefore, we have
\begin{equation}\label{eq:f2} f_{m,n} (t) = \sup_{A \in \cB_m , B \in \cB_n} \frac{\left\| \left[ A^{[n+1,\dots,n+m]} , e^{iH_N t} e^{-i H^{(n)}_N t} B^{[1, \dots, n]} e^{iH_N^{(n)} t} e^{-iH_N t} \right] \right\|}{\| A \| \| B \|} \, . \end{equation} For given $A \in \cB_m$ and $B \in \cB_n$ we define the time-dependent bounded operator acting on $L^2 (\bR^{\nu N})$
\[ g_{A,B} (t) = \left[ A^{[n+1, \dots ,n+m]} , e^{iH_N t} e^{-i H^{(n)}_N t} B^{[1,\dots ,n]} e^{iH_N^{(n)} t} e^{-iH_N t} \right] \] and we compute its time-derivative
\begin{equation*}
\begin{split}
\frac{\rd}{\rd t} g_{A,B} (t) = \; &\left[ A^{[n+1, \dots ,n+m]} , e^{iH_N t} \left[ i (H_N - H_N^{(n)}) , e^{-i H^{(n)}_N t} B^{[1,\dots , n]} e^{iH_N^{(n)} t}\right] e^{-iH_N t} \right] \\ = \; & \left[ A^{[n+1, \dots ,n+m]} , \left[ i \,e^{iH_N t}  (H_N - H_N^{(n)}) e^{-iH_N t} , e^{iH_N t} e^{-i H^{(n)}_N t} B^{([1,\dots ,n]} e^{iH_N^{(n)} t} e^{-iH_N t} \right] \right] \\ = \; & \left[ i \, e^{iH_N t} (H_N - H_N^{(n)}) e^{-iH_N t} , g_{A,B} (t) \right] \\ & + \left[ e^{iH_N t} e^{-i H^{(n)}_N t} B^{[1,\dots,n]} e^{iH_N^{(n)} t}e^{-iH_N t} , \left[ A^{[n+1, \dots,n+m]}, e^{iH_N t} (H_N - H_N^{(n)}) e^{-iH_N t} \right] \right]
\end{split}
\end{equation*}
where, in the last step we used the Jacobi identity. Next,we define
\[ \cH^{(n)} (t) = e^{iH_N t} (H_N - H_N^{(n)}) e^{-iH_N t} \, . \] It is simple to see that $\cH^{(n)} (t)$ generates a two-parameter group of unitary transformations $\cU^{(n)} (t,s)$ satisfying \[ i\partial_t \cU^{(n)} (t,s) = \cH^{(n)} (t) \cU^{(n)} (t,s), \qquad \text{with } \cU^{(n)} (s,s) = 1 \quad \text{for all } s \in \bR \, .\] Therefore, we obtain
\begin{equation}
\begin{split}
\frac{\rd}{\rd t} &\; \cU^{(n)} (0,t) \, g_{A,B} (t) \, \cU^{(n)} (t,0) \\ &=  \cU^{(n)} (0,t) \, \left[ e^{iH_N t} e^{-i H^{(n)}_N t} B^{[1,\dots ,n]} e^{iH_N^{(n)} t}e^{-iH_N t} , \right. \\ &\left. \hspace{3cm} \left[ A^{[n+1, \dots ,n+m]}, e^{iH_N t} (H_N - H_N^{(n)}) e^{-iH_N t} \right] \right] \, \cU^{(n)} (t,0) \, .
\end{split}
\end{equation}
Integrating this identity from time $0$ to time $t$, using that $g_{A,B} (0) = 0$ and the definition of $H^{(n)}_N$, we find that
\begin{equation}
\begin{split}
g_{A,B} (t) = \frac{1}{N} \sum_{\ell =1}^n \sum_{j > n} \int_0^t \rd s \; \cU^{(n)} &(t,s) \left[ e^{iH_N s} e^{-i H^{(n)}_N s} B^{[1,\dots ,n]} e^{iH_N^{(n)} s} e^{-iH_N s} , \right. \\  &\left. \left[ A^{[n+1, \dots, n+m]}, e^{iH_N s} V(x_{\ell} - x_j) e^{-iH_N s} \right] \right] \cU^{(n)} (s,t)\,.
\end{split}\end{equation}
Taking norms, we get
\begin{equation*}
\begin{split}
\| g_{A,B} (t) \| \leq \; &\frac{2}{N} \sum_{\ell=1}^n \sum_{j > n} \int_0^t  \rd s \,\| B \| \,  \left\| \left[ A^{[n+1, \dots ,n+m]}, e^{iH_N s} V(x_{\ell} - x_j) e^{-iH_N s} \right] \right\| \\ \leq \; & \frac{4\, mn \, t \,  \| A \| \, \| B \| \, \| V \|}{N}  + 2n \| B \| \int_0^t \rd s \, \left\| \left[ A^{[n+1, \dots, n+m]} , e^{iH_N s} V (x_1 -x_{n+m+1}) e^{-iH_N s} \right] \right\|
\end{split}
\end{equation*}
where the first term on the last line corresponds to the terms with $j=n+1, \dots , n+m$, while in the second term we used the permutation symmetry. {F}rom (\ref{eq:f1}) and (\ref{eq:f2}), it follows that
\begin{equation}
f_{m,n} (t) \leq \frac{4\, m n \, t \, \| V \|}{N} + 2n \, \| V \| \int_0^t \rd s \; f_{m,2} (s) \, .
\end{equation}
Iterating this equation for $k$ times, we obtain that
\begin{equation}
\begin{split}
f_{m,n} (t) \leq \; &\frac{4\, mn \,\| V \|}{N}  \, t + \frac{4 \, mn \, \| V \|}{N} \sum_{r=1}^{k-1} (4 \| V \|)^{r} \int_0^t \rd s_1 \dots \int_0^{s_{r-1}} \rd s_r \; s_r \\
&+ 2n \| V \|  ( 4 \| V \|)^{k-1} \int_0^t \dots \int_0^{s_{k-1}} \rd s_k \; f_{m,2} (s_r) \,.
\end{split}
\end{equation}
With the a-priori bound $f_{m,2} (s) \leq 2$, it follows that
\[ f_{m,n} (t) \leq \frac{ mn \, }{N} \sum_{r=0}^{k-1} \frac{(4 \| V \| t)^{r+1}}{(r+1)!} + n \frac{(4 \| V \| t)^k}{k!} \leq \frac{mn}{N} \, \left( e^{4 \| V \| t} - 1 \right) + n \frac{(4 \| V \| t)^k}{k!} \]
for all $k \in \bN$. Since the l.h.s. is independent of $k$, we obtain (\ref{eq:claim1}).
\end{proof}

The following proposition is a useful consequence of Theorem \ref{thm:comm}.

\begin{proposition}\label{prop2}
Suppose that $V \in L^{\infty} (\bR^{\nu})$. Denote $\psi_{N,t} = e^{-iH_N t} \psi_N$ the solution to the $N$-particle Schr\"odinger equation with factorized initial data $\psi_N = \ph^{\otimes N}$, for some $\ph \in L^2 (\bR^3)$ with $\| \ph \|_2 = 1$.
Then, for any $A \in \cB_m$, $B \in \cB_n$, we have
\begin{equation}
\begin{split}
\Big| \langle \psi_{N,t} , A^{[i_1, \dots , i_m]} B^{[j_1, \dots , j_n]} \psi_{N,t} \rangle - \langle \psi_{N,t} , A^{[i_1, \dots ,i_m]} \psi_{N,t} \rangle \langle &\psi_{N,t}, B^{[j_1, \dots ,j_n]} \psi_{N,t} \rangle \Big| \\ & \leq \frac{mn \| A \| \| B \|}{N} \, \left(e^{8\| V \| t} -1\right) \,
\end{split}
\end{equation}
for arbitrary integers $1 \leq i_1 < \dots < i_m \leq N$, $1 \leq j_1 < \dots < j_n \leq N$ with $\{ i_1, \dots , i_m \} \cap \{ j_1, \dots ,j_n \} = \emptyset$. In particular, if $\gamma^{(k)}_{N,t}$ denotes the $k$-particle marginal associated with $\psi_{N,t}$, we have
\[ \Big| \tr \; \left(A^{[1,\dots ,m]} \otimes B^{[m+1, \dots ,m+n]} \right) \left( \gamma^{(m+n)}_{N,t} - \gamma^{(m)}_{N,t} \otimes \gamma^{(n)}_{N,t} \right) \Big| \leq \frac{mn \| A \| \| B \|}{N} \, \left(e^{8\| V \| t} -1\right) \, .\]
\end{proposition}

\begin{proof}
Because of the permutation symmetry we have
\begin{multline}
\langle \psi_{N,t} , A^{[i_1, \dots , i_m]} B^{[j_1, \dots , j_n]} \psi_{N,t} \rangle - \langle \psi_{N,t} , A^{[i_1, \dots ,i_m]} \psi_{N,t} \rangle \langle \psi_{N,t}, B^{[j_1, \dots ,j_n]} \psi_{N,t} \rangle \\
= \langle \psi_{N,t} , A^{[1, \dots , m]} B^{[m+1, \dots , m+n]} \psi_{N,t} \rangle - \langle \psi_{N,t} , A^{[1, \dots ,m]} \psi_{N,t} \rangle \langle \psi_{N,t}, B^{[m+1, \dots ,m+n]} \psi_{N,t} \rangle \, .\end{multline}
We observe that
\begin{equation}
\begin{split}
\langle \psi_{N,t} , &A^{[1,\dots,m]} B^{[m+1, \dots,m+n]} \psi_{N,t} \rangle \\ =  \; &\langle \psi_{N} , e^{i H_N t} A^{[1,\dots,m]} e^{-iH_N t} e^{iH_N t} B^{[m+1, \dots,m+n]} e^{-iH_N t} \psi_N \rangle \\ = \; &\langle \psi_{N} , e^{i H_N t} A^{[1,\dots,m]} e^{-iH_N t} |\ph \rangle \langle \ph|^{\otimes N} e^{iH_N t} B^{[m+1, \dots, m+n]} e^{-iH_N t} \psi_N \rangle \\ &+ \sum_{j=1}^N \langle \psi_{N} , e^{i H_N t} A^{[1,\dots,m]} e^{-iH_N t} \, \left[ |\ph \rangle \langle \ph|^{\otimes (j-1)} \otimes \left( 1-|\ph \rangle \langle \ph| \right) \otimes 1^{\otimes (N-j)} \right] \\ & \hspace{5cm} \times  e^{iH_N t} B^{[m+1, \dots, m + n]} e^{-iH_N t} \psi_N \rangle\,.
\end{split}
\end{equation}
Since $\psi_N = \ph^{\otimes N}$, we get
\begin{equation}
\begin{split}
\langle \psi_{N,t} , &A^{[1,\dots,m]} B^{[m+1, \dots, m+n]} \psi_{N,t} \rangle \\= & \; \langle \psi_{N,t} , A^{[1,\dots,m]} \psi_{N,t} \rangle \langle \psi_{N,t},  B^{[m+1, \dots, m+n]} \psi_{N,t} \rangle \\ &+ \sum_{j=1}^m \left\langle \psi_{N} , e^{i H_N t} A^{[1, \dots ,m]} e^{-iH_N t} \, \left( |\ph \rangle \langle \ph|^{\otimes (j-1)} \otimes 1^{\otimes (N-j+1)}\right)  \right. \\ & \left. \hspace{4cm} \times  \left[ \left( 1-|\ph \rangle \langle \ph| \right)^{[j]}, e^{iH_N t} B^{[m+1, \dots, m+n]} e^{-iH_N t} \right] \psi_N \right\rangle \\ &+ \sum_{j=m+1}^{m+n} \left\langle \psi_{N} , \left[ e^{i H_N t} A^{[1,\dots,m]} e^{-iH_N t} , (1-|\ph\rangle \langle \ph|)^{[j]} \right] \, \left(\ |\ph \rangle \langle \ph|^{\otimes (j-1)} \otimes 1^{\otimes (N-j+1)} \right)  \right. \\ &\left. \hspace{4cm} \times  \; e^{iH_N t} B^{[m+1, \dots ,m+n]} e^{-iH_N t} \psi_N \right\rangle
\\ &+ \sum_{j=m+n+1}^N \left\langle \psi_{N} , \left[ e^{i H_N t} A^{[1,\dots ,m]} e^{-iH_N t}, \left( 1-|\ph \rangle \langle \ph| \right)^{[j]} \right] \, \left( |\ph \rangle \langle \ph|^{\otimes (j-1)} \otimes 1^{\otimes (N-j+1)} \right)\, \right. \\ & \left.\hspace{4cm} \times \left[ \left( 1-|\ph \rangle \langle \ph| \right)^{[j]} , \, e^{iH_N t} B^{[m+1, \dots, m+n]} e^{-iH_N t} \right] \psi_N \right\rangle \,,
\end{split}
\end{equation}
where we used the notation $(1-|\ph \rangle \langle \ph|)^{[j]} = 1^{(j-1)} \otimes (1-|\ph \rangle \langle \ph|) \otimes 1^{(N-j-1)}$ for the operator acting as the projection $(1-|\ph \rangle \langle \ph|)$ over the $j$-th particle, and as the identity over the other $(N-1)$ particles. Note that, by definition $((1-|\ph\rangle \langle \ph|)^{[j]} )^2 = 1-|\ph \rangle \langle \ph|$. {F}rom Theorem \ref{thm:comm}, we obtain
\begin{equation*}
\begin{split}
\Big| \langle \psi_{N,t} , A^{[1, \dots,m]} B^{[m+1, \dots ,m+n]} \psi_{N,t} \rangle - &\langle \psi_{N,t} , A^{[1,\dots,m]} \psi_{N,t} \rangle \langle \psi_{N,t},  B^{[m+1, \dots , m+n]} \psi_{N,t} \rangle \Big| \\ \leq \; & \frac{2\, mn \, \| A \| \| B \|}{N} \, (e^{4\| V \| t} -1) + \frac{mn \| A \| \| B \|}{N} (e^{4\| V \| t} -1 )^2 \\ = \; &\frac{mn \| A \| \|B \|}{N} \, (e^{8\| V \| t}  -1)
\end{split}
\end{equation*}
which completes the proof of the proposition.
\end{proof}

\section{Derivation of the Hartree equation}\label{sec:proof}

Using the bounds on the correlations obtained in the previous section, we can now state and prove our main result.

\begin{proof}[Proof of Theorem \ref{thm:main}]
{F}rom the BBGKY hierarchy, we obtain the integral equation
\begin{equation}\begin{split} \gamma_{N,t}^{(k)} = \; &\cU^{(k)} (t) |\ph \rangle \langle \ph|^{\otimes k} - \frac{i}{N} \sum_{i<j}^k \int_0^t \rd s \, \cU^{(k)}(t-s) \, \left[ V(x_i - x_j) , \gamma_{N,s}^{(k)} \right] \\ &-i \, \frac{(N-k)}{N} \sum_{j=1}^k \int_0^t \rd s \, \cU^{(k)} (t-s)\, \tr_{k+1} \; \left[ V (x_j -x_{k+1}) , \gamma_{N,s}^{(k+1)} \right] \, , \end{split}\end{equation}
where we defined the free evolution
\[ \cU^{(k)} (t) \gamma^{(k)} = e^{-i\sum_{j=1}^k (-\Delta_j + U (x_j)) t} \;  \gamma^{(k)} \; e^{i \sum_{j=1}^k (-\Delta_j + U (x_j)) t} \, , \] and where $\tr_{k+1}$ denotes the partial trace over the $(k+1)$-th particle. On the other hand, it is simple to check that, if $\ph_t$ is a solution to the nonlinear Hartree equation \[ i \partial_t \ph_t = - \Delta \ph_t + \left( V * |\ph_t|^2 \right) \ph_t \] then the orthogonal projection $|\ph_t \rangle \langle \ph_t|^{\otimes k}$ satisfies the equation
\[ |\ph_t \rangle \langle \ph_t |^{\otimes k} = \cU^{(k)} (t) |\ph \rangle \langle \ph|^{\otimes k} -i \sum_{j=1}^k \int_0^t \rd s \, \cU^{(k)} (t-s) \tr_{k+1} \; \left[ V (x_j -x_{k+1}) , |\ph_t \rangle \langle \ph_t|^{\otimes (k+1)} \right] \, . \]
Therefore, for an arbitrary $k$-particle observable $J^{(k)}$ over $L^2 (\bR^{k \nu})$, we obtain
\begin{equation}\label{eq:bd1}
\begin{split}
\tr J^{(k)} \Big( \gamma_{N,t}^{(k)} - &|\ph_t \rangle \langle \ph_t|^{\otimes k} \Big) \\ = &
- \frac{i}{N} \sum_{i <j}^k \int_0^t \rd s \; \tr \, \left(\cU^{(k)} (s-t) J^{(k)} \right) \; \left[ V (x_i -x_j) , \gamma^{(k)}_{N,s} \right] \\ &+ i \frac{k}{N} \sum_{j=1}^k \int_0^t \rd s \;  \tr \, \left(\cU^{(k)} (s-t) J^{(k)} \right) \; \left[ V (x_j -x_{k+1}) , \gamma^{(k+1)}_{N,s} \right] \\
&-i \sum_{j=1}^k \int_0^t \rd s \; \tr \,
\left(\cU^{(k)} (s-t) J^{(k)} \right) \; \left[ V (x_j -x_{k+1}) , \left( \gamma^{(k+1)}_{N,s} - \gamma^{(k)}_{N,s} \otimes \gamma^{(1)}_{N,s} \right) \right] \\
&-i\sum_{j=1}^k \int_0^t \rd s \; \tr \,
\left(\cU^{(k)} (s-t) J^{(k)} \right) \; \left[ V (x_j -x_{k+1}) , \left( \gamma^{(k)}_{N,s} - |\ph_s \rangle \langle \ph_s|^{\otimes k} \right) \otimes |\ph_s \rangle \langle \ph_s| \right] \\
&-i\sum_{j=1}^k \int_0^t \rd s \; \tr \,
\left(\cU^{(k)} (s-t) J^{(k)} \right) \; \left[ V (x_j -x_{k+1}) , \gamma_{N,s}^{(k)} \otimes \left( \gamma^{(1)}_{N,s} - |\ph_s \rangle \langle \ph_s | \right) \right] \, .
\end{split}
\end{equation}
We bound the fourth and fifth term using that
\begin{equation}
\begin{split}\Big| \tr \; \left( \cU^{(k)} (s-t) J^{(k)} \right) \Big[ V(x_j - x_{k+1}), \Big( &\gamma^{(k)}_{N,s} - |\ph_s \rangle \langle \ph_s|^{\otimes k} \Big) \otimes |\ph_s \rangle \langle \ph_s| \Big] \Big| \\ &\leq 2 \| V \| \, \| J^{(k)} \| \, \tr \, \left| \gamma^{(k)}_{N,s} - |\ph_s \rangle \langle  \ph_s|^{\otimes k} \right|
\end{split}
\end{equation}
and that, analogously,
\begin{equation}
\begin{split}
\Big| \tr \; \left( \cU^{(k)} (s-t) J^{(k)} \right) \Big[ V(x_j - x_{k+1}), &\gamma_{N,s}^{(k)} \otimes \left( \gamma^{(1)}_{N,s} - |\ph_s \rangle \langle \ph_s| \right) \Big] \Big| \\ &\leq 2 \| V \| \, \| J^{(k)} \| \, \tr \, \left| \gamma^{(1)}_{N,s} - |\ph_s \rangle \langle  \ph_s| \right|\,.
\end{split}
\end{equation}
As for the third term on the r.h.s. of (\ref{eq:bd1}), we expand the potential in a Fourier integral. We obtain
\begin{equation}
\begin{split}
\Big| \tr \; &\left(\cU^{(k)} (s-t) J^{(k)} \right) \, \left[ V (x_j -x_{k+1}), \left(\gamma_{N,s}^{(k+1)} - \gamma_{N,s}^{(k)} \otimes \gamma_{N,s}^{(1)} \right) \right] \Big| \\  & \leq \, \int \rd q \, |\widehat{V} (q)| \; \Big| \tr \; \left(\cU^{(k)} (s-t) J^{(k)} \right) \, \left[ e^{i q \cdot x_j} \otimes e^{-i q \cdot x_{k+1}}, \left(\gamma_{N,s}^{(k+1)} - \gamma_{N,s}^{(k)} \otimes \gamma_{N,s}^{(1)} \right) \right] \Big| \\
&\leq \, \int \rd q \, |\widehat{V} (q)| \; \Big| \tr \; \left(\cU^{(k)} (s-t) J^{(k)} \right) e^{i q \cdot x_j} \otimes e^{-i q \cdot x_{k+1}} \left(\gamma_{N,s}^{(k+1)} -\gamma_{N,s}^{(k)} \otimes \gamma_{N,s}^{(1)} \right) \Big| \\ &\hspace{.5cm} + \int \rd q \, |\widehat{V} (q)| \; \Big| \tr \;  \, e^{i q \cdot x_j}\left(\cU^{(k)} (s-t) J^{(k)} \right) \otimes e^{-i q \cdot x_{k+1}} \left(\gamma_{N,s}^{(k+1)} - \gamma_{N,s}^{(k)} \otimes \gamma_{N,s}^{(1)} \right)  \Big| \\
&\leq \,\frac{2 k\, \| \widehat{V} \|_1 \, \| J^{(k)} \|}{N} \; \left( e^{8 \| V \| s} -1\right)
\end{split}
\end{equation}
where we applied Proposition \ref{prop2}. Thus, from (\ref{eq:bd1}), we find that
\begin{equation}
\begin{split}
\Big| \tr \; J^{(k)} &\left( \gamma_{N,t}^{(k)} - |\ph_t \rangle \langle \ph_t|^{\otimes k} \right) \Big| \\ \leq \; &2k \, \| V \| \| J^{(k)} \| \int_0^t \rd s \, \tr \, \left|  \gamma_{N,s}^{(k)} - |\ph_s \rangle \langle \ph_s|^{\otimes k} \right| + 2k \, \| V \| \| J^{(k)} \| \int_0^t \rd s \, \tr\, \left|  \gamma_{N,s}^{(1)} - |\ph_s \rangle \langle \ph_s| \right|\\
&+ \frac{k^2 \, \| \widehat{V} \|_1 \,  \| J^{(k)} \|}{4 \| V \| N} \; \left(e^{8 \| V \| t} - 1 \right) + \frac{3 t k^2 \| V \| \| J^{(k)} \|}{N} \, .
\end{split}
\end{equation}
 For a trace class operator $A$ over $L^2 (\bR^{\nu k})$, let $\| A \|_1 = \tr \, |A|$ be the trace norm of $A$. If we denote by $\cK_k$ the algebra of compact operators on $L^2 (\bR^{\nu k})$, we have the variational characterization of the trace norm
 \[ \| A \|_1 = \tr \; \left|  A \right| = \sup_{J \in \cK_k \, , \, \| J \| = 1} \left| \tr \; J \, A \right| \, .\] Therefore we find
\begin{equation}\label{eq:bd2}
\begin{split}
\left\| \gamma_{N,t}^{(k)} - |\ph_t \rangle \langle \ph_t|^{\otimes k} \right\|_1 \leq \; &2k \| V \| \int_0^t \rd s \, \left( \left\| \gamma_{N,s}^{(k)} - |\ph_s \rangle \langle \ph_s|^{\otimes k} \right\|_1 + \left\| \gamma_{N,s}^{(1)} - |\ph_s \rangle \langle \ph_s| \right\|_1 \right) \\ &+ \frac{k^2 \, \lambda_V}{2N} \, \left( e^{8\| V \| t} -1 \right) .
\end{split}
\end{equation}
with $\lambda_V = 1 + (\| \widehat{V} \|_1 / \| V \|)$. For $k=1$, we obtain
\[ \left\| \gamma_{N,t}^{(1)} - |\ph_t \rangle \langle \ph_t| \right\|_1 \leq \; 4 \| V \| \int_0^t \rd s \,  \left\| \gamma_{N,s}^{(1)} - |\ph_s \rangle \langle \ph_s| \right\|_1 + \frac{\lambda_V}{2N} \, \left( e^{8\| V \| t} -1 \right)\,. \]
Iterating this inequality $n$ times, it follows that
\begin{equation}
\begin{split}
\left\| \gamma_{N,t}^{(1)} - |\ph_t \rangle \langle \ph_t| \right\|_1 \leq \; & \frac{\lambda_V}{2N} \, \left( e^{8 \| V \| t} -1 \right)  \\ & + \frac{\lambda_V}{2N} \sum_{m=1}^{n-1} (4 \| V \| )^{m}  \int_0^t \rd s_1 \int_0^{s_1} \rd s_2 \dots \int_0^{s_{m-1}} \rd s_m \,  \, \left( e^{8\| V \| s_m} -1 \right)  \\ &+ (4 \| V \| )^{n} \int_0^t \rd s_1 \dots \int_0^{s_{n-1}} \rd s_n \; \left\| \gamma_{N,s_n}^{(1)} - |\ph_{s_n} \rangle \langle \ph_{s_n} | \right\|_1 \, .
\end{split}
\end{equation}
Since \[ \int_0^t \rd s_1 \dots \int_0^{s_{m-1}} \rd s_m \; e^{8\| V \| s_m} \leq \frac{\left(e^{8 \| V \| t} -1 \right)}{(8 \| V \|)^m} \] and $\| \gamma_{N,s_n}^{(1)} - |\ph_{s_n} \rangle \langle \ph_{s_n} | \|_1 \leq 2$, we find
\begin{equation}
\left\| \gamma_{N,t}^{(1)} - |\ph_t \rangle \langle \ph_t| \right\|_1 \leq \frac{\lambda_V}{N} \, \left( e^{8 \| V \| t} -1 \right) \, .
\end{equation}
Inserting this bound into the r.h.s. of (\ref{eq:bd2}), we conclude that
\[  \left\| \gamma_{N,t}^{(k)} - |\ph_t \rangle \langle \ph_t|^{\otimes k} \right\|_1 \leq 2k \| V \| \int_0^t \rd s \, \left\| \gamma_{N,s}^{(k)} - |\ph_s \rangle \langle \ph_s|^{\otimes k} \right\|_1 + \frac{k^2 \, \lambda_V}{N} \,
\left( e^{8\| V \| t} -1 \right) \, .\]
Iteration leads to
\begin{equation}\begin{split}
\left\| \gamma_{N,t}^{(k)} - |\ph_t \rangle \langle \ph_t|^{\otimes k} \right\|_1 \leq \; & \frac{k^2 \lambda_V}{N} \, \left( e^{8 \| V \| t} -1 \right)  \\ &+ \frac{k^2 \lambda_V}{N} \sum_{m=1}^{n-1} (2k \| V \| )^{m}  \int_0^t \rd s_1 \int_0^{s_1} \rd s_2 \dots \int_0^{s_{m-1}} \rd s_m \,  \, \left(e^{8\| V \| s_m}-1\right)  \\ &+ (2k \| V \| )^{n} \int_0^t \rd s_1 \dots \int_0^{s_{n-1}} \rd s_n \; \left\| \gamma_{N,s_n}^{(k)} - |\ph_{s_n} \rangle \langle \ph_{s_n} |^{\otimes k} \right\|_1
\end{split}
\end{equation}
which implies that
\[ \left\| \gamma_{N,t}^{(k)} - |\ph_t \rangle \langle \ph_t|^{\otimes k} \right\|_1 \leq \frac{k^2 \lambda_V}{N} \; e^{2k \| V \| } \, \left( e^{8 \| V\| t}-1 \right) \, .\]
\end{proof}

\end{document}